\documentclass[onecolumn,draftcls]{IEEEtran}

\usepackage[T1]{fontenc}

\usepackage[utf8]{inputenc}
\usepackage[english]{babel}
\usepackage{csquotes}
\usepackage{amsmath, amssymb, amsfonts, amsthm}

\usepackage{braket}
\usepackage{graphicx}

\usepackage{xcolor}
\definecolor{URLColorBlue}{HTML}{2A1B81}
\definecolor{BlueGreen}{cmyk}{0.85,0,0.33,0}
\definecolor{RawSienna}{HTML}{A0522D}
\definecolor{TealLight}{HTML}{00C798}
\definecolor{Teal}{HTML}{008080}

\usepackage{hyperref}

\usepackage{array,multirow}

\hypersetup{colorlinks,
			citecolor=TealLight,
			urlcolor=URLColorBlue,
			pdftitle={optimal_systematic},
			pdfauthor={Eleonora Guerrini, Alessio Meneghetti, Massimiliano Sala},
			pdfkeywords={Coding Theory} {optimal codes},
}

\theoremstyle{plain}
\newtheorem{thm}{Theorem}
\newtheorem{lem}[thm]{Lemma}
\newtheorem{prop}[thm]{Proposition}
\newtheorem{cor}[thm]{Corollary}

\theoremstyle{definition}
\newtheorem{defn}[thm]{Definition}

\theoremstyle{remark}
\newtheorem{rem}[thm]{Remark}

\newcommand{\FF}{\mathbb{F}}

\newcommand{\dd}{\mathrm{d}}

\newcommand{\len}{\mathrm{len}}

\title{On optimal nonlinear systematic codes}
\author{
\IEEEauthorblockN{Eleonora Guerrini\IEEEauthorrefmark{1},  
Alessio Meneghetti\IEEEauthorrefmark{2}
 and Massimiliano Sala\IEEEauthorrefmark{2}}\thanks{This paper was presented in part at the Ninth International Workshop on Coding and Cryptography, WCC 2015, April 13-17, 2015 Paris, France.}
 
\IEEEauthorblockA{\IEEEauthorrefmark{1}LIRMM, Universit\'e
de Montpellier 2, France}

\IEEEauthorblockA{\IEEEauthorrefmark{2}Department of Mathematics,
University of Trento, Italy}
}

\begin{document}
\maketitle
\begin{abstract}
Most bounds on the size of codes hold for any code, whether linear or not. 
Notably, the Griesmer bound holds only in the linear case and so optimal linear codes are not necessarily optimal codes.
In this paper we identify code parameters $(q,d,k)$, namely field size, minimum distance and combinatorial dimension, for which the Griesmer bound holds also in the (systematic) nonlinear case. 
Moreover, we show that the Griesmer bound does not necessarily hold for a systematic code
 by explicit construction of a family of optimal systematic binary codes.
On the other hand, we are able to provide some versions of the Griesmer bound holding for all systematic codes.
\end{abstract}


\section{Introduction}
In this work we consider three sets of codes: linear, systematic and nonlinear codes. With code $C$ we mean a set of $M$ vectors in the vector space $\left(\mathbb{F}_q\right)^n$, where $\mathbb{F}_q$ is the finite field with $q$ elements. We refer to each of these vectors as a \textit{codeword} $c\in C$, to $n$ as the \textit{length} of $C$ and to $M$ as its \textit{size}. We denote with $d$ the minimum distance of $C$, i.e. the minimum among the Hamming distances between any two distinct codewords in $C$. A code $C$ with such parameters is denoted by an $(n,M,d)_q$ code.  
$C$ is a \textit{linear} code if $C$ is a vector subspace of $\left(\mathbb{F}_q\right)^n$. 
In this case, $M=q^k$ for a certain positive integer $k$ called the \textit{dimension} of the code.
A code which is not equivalent to any linear code is called a \textit{strictly nonlinear} code.
\\
Systematic codes form an important family of nonlinear codes. As we will show in Section \ref{sec:counterexample}, systematic codes can achieve better error correction capability than any linear code with the same parameters. On the other hand, due to their particular structure, systematic codes can achieve faster encoding and decoding procedures than nonlinear non-systematic codes. Moreover, many known families of optimal codes are systematic codes (see e.g., \cite{CGC-cd-art-preparata}, \cite{CGC-kerd72}).

\begin{defn}\label{defn: classical systematic}
An $(n,q^k,d)_q$ systematic code $C$ is the image of an injective map $F:\left(\FF_q\right)^k\to \left(\FF_q\right)^n$, $n\ge k$, s.t. 
a~vector $X = (x_1,\ldots,x_k) \in (\FF_q)^k$ is mapped to a vector $$(x_1,\ldots,x_k,f_{k+1}(X),\ldots,f_n(X)) \in (\FF_q)^n,$$
where 
$f_i,i=k+1,\ldots,n$ are maps from $(\FF_q)^k$ to $\FF_q$. 
We refer to $k$ as the \textit{combinatorial dimension} of $C$. 
The coordinates from 1 to $k$ are called \textit{systematic}, while those from $k+1$ to $n$ are called \textit{non-systematic}.
\end{defn}
It is well known that any linear code is equivalent to a systematic one. Note that $C$ is linear if and only if the maps $f_i$ are linear.

Recent results on systematic codes can be found in \cite{CGC-cod-art-alderson2008maximal} and \cite{CGC-cod-art-alderson2009maximality}, where it is proved that if a linear code admits  an extension (both the length and the distance are increased exactly by $1$), then it admits also a linear extension. Therefore, we observe that if puncturing a systematic code $C$ we obtain a linear code, then there exists a linear code with the same parameters as $C$. 
We denote with $\len(C),\dim(C),\dd(C)$, respectively, the length, the (combinatorial) dimension and the minimum distance of a code $C$.\\
A classical problem in coding theory is to determine the parameters of optimal codes, and this characterization is usually carried on by presenting bounds on the minimum distance, on the size, or on the length of codes. Since two equivalent codes have the same parameters, we can always assume that the zero codeword belongs to $C$.
In this work we consider the following definition of an optimal code.
\begin{defn}
Let $k$ and $d$ be two positive integers. An $(n,M,d)_q$ code $C$ is \textit{optimal} if all codes with the same distance and size  have length at least $n$.\\
An $(n,q^k,d)_q$ systematic code $C$ is \textit{optimal} if all systematic codes with the same distance and dimension  have length at least $n$.\\
We denote with $N_q(M,d)$, $S_q(k,d)$ and $L_q(k,d)$ the minimum length of, respectively, a nonlinear, systematic and linear code.
\end{defn}
We are interested in analysing the minimum possible length of a code whose distance and size are known.
\begin{rem}\label{rem: inequalities}
Clearly, $N_q(q^k,d)\leq S_q(k,d) \leq L_q(k,d)$.
\end{rem}
A well-known bound on the size of binary codes is the Plotkin bound \cite{CGC-cd-art-plotk60}, which can be applied to any code whose minimum distance is large enough w.r.t. its length.
\begin{thm}[Plotkin bound]\label{thm: plotkin nonbinary}
Any $(n,M,d)_q$ code satisfies
\begin{equation}\label{eq: plotkin length}
n \ge \left\lceil d \left(\frac{1-\frac{1}{M}}{1-\frac{1}{q}}\right)\right\rceil .
\end{equation}
Moreover, any $(n,M,d)_q$ code such that $n< \frac{qd}{q-1}$ satisfies 
$$
M \leq \left\lfloor \frac{d}{d-\left(1-\frac{1}{q}\right)n} \right\rfloor .
$$
\end{thm}
We also recall another useful bound, which is known to hold only for linear codes.
\begin{thm}[Griesmer bound]
Let $k$ and $d$ be two positive integers. Then
\begin{equation}\label{eq: griesmer}
L_q(k,d)\ge g_q(k,d) := \sum_{i=0}^{k-1}\left\lceil\frac{d}{q^i}\right\rceil
\end{equation}
\end{thm} 
\noindent 
The Griesmer bound, which can be seen as an extension of the Singleton bound \cite[Section 2.4]{CGC-cd-book-huffmanPless03} in the linear case, 
was introduced by Griesmer \cite{CGC-cd-art-griesm60} in the case of binary linear codes and then 
generalized by Solomon and Stiffler \cite{CGC-cod-art-solstiffl65} in the case of $q$-ary linear codes. 
It is known that the Griesmer bound is not always sharp \cite{CGC-cod-art-maruta1996non}, \cite{CGC-cod-art-van1980uniqueness}, \cite{CGC-cod-art-Maruta1997}. \\
Important examples of linear codes meeting the Griesmer bound are the simplex code \cite[Section 1.3]{CGC-cd-book-huffmanPless03} and the $[11,5,6]_3$ Golay code \cite[Section 1.12]{CGC-cd-book-huffmanPless03}, \cite{CGC-gola49}. \\

Many papers, such as
\cite{CGC-cod-art-helleseth1981characterization}, 
\cite{CGC-cod-art-hamada1993characterization}, 
\cite{CGC-cod-art-tamari1984linear}, 
\cite{CGC-cod-art-Maruta1997}, and 
\cite{CGC-cod-art-Klein04}, 
have characterized classes of linear codes meeting the Griesmer bound.
In particular, finite projective geometries play an important role in the study of these codes. 
For example in 
\cite{CGC-cod-art-hellesth1992projective}, 
\cite{CGC-cod-art-Hamada93} and 
\cite{CGC-cod-art-Tamari93} 
minihypers and maxhypers are used to characterize linear codes meeting the Griesmer bound. 
Research has been done also to characterize the codewords of linear codes meeting the Griesmer bound \cite{CGC-cod-art-ward1998divisibility}.\\
Many known bounds on the size of codes, for example 
the Johnson bound \cite{CGC-cd-art-john62},\cite{CGC-cd-art-john71},\cite{CGC-cd-book-huffmanPless03}, 
the Elias-Bassalygo bound \cite{CGC-cd-art-bass65},\cite{CGC-cd-book-huffmanPless03}, 
the Hamming (Sphere Packing) bound, 
the Singleton bound \cite{CGC-cd-book-pless98},  
the Zinoviev-Litsyn-Laihonen bound \cite{CGC-cod-art-zinlyts1984shortening}, \cite{CGC-cod-art-litlai98},
the Bellini-Guerrini-Sala bound \cite{CGC-cd-art-BellGuerrSal2014}, and
the Linear Programming bound \cite{CGC-cd-art-Dels73},
are true for both linear and (systematic) nonlinear codes. \\
On the other hand, the proof of the Griesmer bound heavily relies on the linearity of the code and it cannot be applied to all codes. 
\\ 

In this paper we present our results on  systematic codes and their relations to (possible extensions of) the Griesmer bound. In Section \ref{sec:systematic_nonlinear_codes} we prove that,
once $q$ and $d$ have been chosen, if all nonlinear $(n,q^{k},d)_q$ systematic codes with $k<1+\log_qd$ respect the Griesmer bound, then the Griesmer bound holds for all systematic codes with the same $q$ and $d$. Therefore, for any $q$ and $d$ only a finite set of $(k,n)$ pairs has to be analysed in order to prove the bound for all $k$ and $n$. In Section \ref{sec:systematic_griesmer} we identify several families of parameters for which the Griesmer bound holds in the systematic (nonlinear) case. 
In Section \ref{sec:weak_griesmer} we provide some versions of the Griesmer bound holding for systematic codes. 
\\
In the next sections we study optimal binary codes with small size, namely $M=4$ and $M=8$. In Section \ref{sec: optimal 4 words} we show that all optimal binary codes with $4$ codewords are necessarily (equivalent to) linear codes. In Section \ref{sec: code structure 8 codewords}  we show that for any possible distance, there exist binary linear codes with $8$ codewords achieving the Plotkin bound, and this implies that $N_2(8,d)=S_2(3,d)=L_2(3,d)$. Finally, in Section \ref{sec:counterexample}, we show explicit counterexamples of binary systematic codes for which the Griesmer bound does not hold, by constructing a family of optimal binary systematic codes. 
In the final section we draw our conclusions and hint at a future work and open problems.

From now on, $n$, $k$ and $d$ are positive integers, $n>k$, and $q\ge 2$ is the power of a prime.


\section{A sufficient condition to prove the Griesmer bound for systematic codes}\label{sec:systematic_nonlinear_codes}
The following proposition and lemma are well-known, we however provide a sketch of their proofs because they anticipate our later argument.
\begin{prop}\label{prop: shortening systematic codes}
Let $C$ be an $(n,q^k,d)$ systematic code, and $C'$ be the code obtained by shortening $C$ in a systematic coordinate. Then $C'$ is an $(n-1,q^{k-1},d')$ systematic code with $d'\ge d$.
\end{prop}
\begin{proof}
To obtain $C'$, consider the code $C'' = \left\{F(X) \mid X = (0,x_2,\ldots,x_k) \in \left(\FF_q\right)^k \right\}$, i.e. the subcode of $C$ which is the image of the set of messages whose first coordinate is equal to $0$. 
Then $C''$ is such that $\dim(C'') = k-1$ and $\dd(C'') \ge d$. 
Since, by construction, all codewords have the first coordinate equal to zero, we obtain the code $C'$ by puncturing $C''$ on the first coordinate, so that $\len(C')=n-1$ and $d' = \dd(C') = \dd(C'') \ge d$.
\end{proof}
\begin{lem}\label{lem: reducing d}
For any $(n,q^k,d)$ systematic code $C$, there exists an $(n,q^k,\bar{d})$ systematic code $\bar{C}$ for any $1\leq\bar{d}\leq d$. 
\end{lem}
\begin{proof}
Since $n > k$, we can consider the code $C^1$ obtained by puncturing $C$ in a non-systematic coordinate.  
$C^1$ is an $(n-1,q^k,d^{(1)})$ systematic code. Of course, either $d^{(1)}=d$ or $d^{(1)}=d-1$.\\
By puncturing at most $n-k$ non-systematic coordinates, we will find a code whose distance is $1$. Then there must exists an $i\leq n-k$ such that the code $C^i$, obtained by puncturing $C$ in the last $i$ coordinates, has distance equal to $\bar{d}$. Once the $(n-i,q^k,\bar{d})$ code $C^i$ has been found, we can obtain the claimed code $\bar{C}$ by padding $i$ zeros to all codewords in $C^i$.
\end{proof}
%
We are ready to present our first result.
\begin{thm}\label{thm: minimum k q}
For fixed $q$ and $d$, if
\begin{equation} \label{eq:griesmer_inequality}
 S_q(k,d) \ge g_q(k,d)
\end{equation}
for all $k$ such that $1 \le k < 1+\log_q d$, 
then \eqref{eq:griesmer_inequality} holds for any $k$,
i.e. the Griesmer bound is true for all systematic codes over $\FF_q$ with minimum distance $d$.
\end{thm}
Before proving it, we remark that an equivalent formulation for Theorem \ref{thm: minimum k q} could be: \textit{If there exists an $(n,q^k,d)_q$ systematic code which does not satisfy the Griesmer bound, then there exists an $(n',q^{k'},d)_q$ systematic code with $k'<1+\log_qd$ which does not satisfy the Griesmer bound.}
\begin{proof}
 For each fixed $d$ and $q$, suppose there exists an $(n,q^k,d)_q$ systematic code not satisfying the Griesmer bound, i.e., there exists $k$ such that $S_q(k,d) < g_q(k,d)$. 
 Let us call 
 $\Lambda_{q,d} = \{k \ge 1 \mid S_q(k,d) < g_q(k,d) \}$.\\
 If $\Lambda_{q,d}$ is empty then the Griesmer bound is true for such parameters $q,d$.\\
 Otherwise, there exists a minimum $k'\in\Lambda_{q,d}$ such that $S_q(k',d) < g_q(k',d)$.\\
 In this case we can consider an $(n,q^{k'},d)_q$ systematic code $C$ not verifying the Griesmer bound,
 $n = S_q(k',d)$.
\\
We obtain an $(n-1,q^{k'-1},d')$ systematic code $C'$ whose distance is $d'\ge d$ by applying Proposition \ref{prop: shortening systematic codes} to $C$, then we apply Lemma \ref{lem: reducing d} to $C'$, hence we obtain an $(n-1,q^{k'-1},d)_q$ systematic code $\bar{C}$. 
\\
 Since $k'$ was the minimum among all the values in $\Lambda_{q,d}$, then the Griesmer bound holds for $\bar{C}$, and so 
 \begin{equation}\label{eq: griesmer n-1}
 n-1 \ge g_q(k'-1,d) = \sum_{i=0}^{k'-2}\left\lceil\frac{d}{q^i}\right\rceil.
 \end{equation}
 We observe that, if $q^{k'-1}\ge d$, then $\left\lceil\frac{d}{q^{k'-1}}\right\rceil=1$, so we can rewrite \eqref{eq: griesmer n-1} as
 \begin{align*}
  n \ge \sum_{i=0}^{k'-2}\left\lceil\frac{d}{q^i}\right\rceil +1\ge \sum_{i=0}^{k'-2}\left\lceil\frac{d}{q^i}\right\rceil + \left\lceil\frac{d}{q^{k'-1}}\right\rceil 
      = \sum_{i=0}^{k'-1}\left\lceil\frac{d}{q^i}\right\rceil
      = g_q(k',d)
 \end{align*}
 Since we supposed $n < g_q(k',d)$, we have reached a contradiction with the assumption $q^{k'-1}\ge d$. Hence for such $d$, the minimum $k$ in $\Lambda_{q,d}$ must satisfy $q^{k-1}<d$, which is equivalent to our claimed expression $k<1+\log_q d$.
\end{proof}
%
\section{Some parameters for which the Griesmer bound holds in the systematic case}
\label{sec:systematic_griesmer}
In this section we identify several sets of parameters $(q,d)$ for which the Griesmer bound holds for systematic codes. 
Subsections \ref{subsec: d|2q} and  \ref{section: q^(d-1)|d} deal with $q$-ary codes, while in Subsection \ref{subsec: 2r2s} we consider the special case of binary codes.
\subsection{The case $d \le 2q$}\label{subsec: d|2q}
\begin{thm}\label{thm: griesmer 2q}
If $d\le 2q$ then $S_q(k,d)\ge g_q(k,d)$.
\end{thm}
\begin{proof}
First, consider the case $d \leq q$. 
By Theorem \ref{thm: minimum k q} it is sufficient to show that, fixing $q$ and $d$, for any $n$ there is no $(n,q^k,d)_q$ systematic code with $1 \le k<1+\log_q d$ and $n < g_q(k,d)$.
If $1 \le k<1+\log_q d$ then $\log_q d \le \log_q q = 1$, and so $k$ may only be 1. 
Since $g_q(1,d) = d$ and $n \ge d$, we clearly have that $n \ge g_q(1,d)$.

Now consider the case $q<d\leq 2q$. 
If $1 \le k<1+\log_q d$ then $\log_q d \le \log_q 2q = 1+\log_q 2$, and so $k$ can only be 1 or 2. 
We have already seen that if $k=1$ then $n \ge g_q(k,d)$ for any $n$, so suppose $k=2$. 
If an $(n,q^2,d)_q$ systematic code $C$ exists with $n<\sum_{i=0}^1\left\lceil\frac{d}{q^i}\right\rceil=d+2$, then by the Singleton bound we can only have $n=d+1$.  
Therefore $C$ must have parameters $(d+1,q^2,d)$.
In \cite[Ch.~10]{CGC-cod-book-hill1986first} it is proved that a $q$-ary $(n,q^2,n-1)_q$ code is equivalent to a set of $n-2$ mutually orthogonal Latin squares (MOLS) of order $q$,
and that there are at most $q-1$ Latin squares in any set of MOLS of order $q$ (Theorem 10.18). In our case $n=d+1>q+1$, therefore $n-2>q-1$. The existence of $C$ would imply the existence of a set of more than $q-1$ MOLS, which is impossible.
\end{proof}
\subsection{The case $q^{k-1} \mid d$}
 \label{section: q^(d-1)|d}
The following proposition is a simple consequence of the Plotkin bound that implies some results on values for the distance and dimension for which the Griesmer bound holds in the nonlinear case. We will also make use of this result to obtain a version of the Griesmer bound which can be applied to all systematic codes.
\begin{prop}\label{prop: plotkin equal griesmer}
If $q^{k-1}\mid d$, then the Griesmer bound coincides with the Plotkin bound in equation \eqref{eq: plotkin length}.
\end{prop}
\begin{proof}
If $q^{k-1}\mid d$, then
$
g_q(k,d)=\sum_{i=0}^{k-1}\frac{d}{q^{i}}=d\sum_{i=0}^{k-1}\frac{1}{q^i}=d\frac{1-\frac{1}{q^k}}{1-\frac{1}{q}}
$.
\end{proof}

\begin{cor}\label{cor: griesmer 1}
Let $r\ge 1$, then $N_q(q^k,q^{k-1}r)\ge g_q(k,q^{k-1}r)$.
\end{cor}
\begin{proof}
Follows directly from Proposition \ref{prop: plotkin equal griesmer}.
\end{proof}
Note that Corollary  \ref{cor: griesmer 1} is not restricted to systematic codes, and holds for any code with at least $q^k$ codewords, so we can obtain directly the next corollary.
\begin{cor}
Let $M\ge q^k$ and $r\ge 1$, then $N_q(M,q^{k-1}r)\ge g_q(k,q^{k-1}r)$.
\end{cor}

The following lemma holds for any nonlinear code.
\begin{lem}\label{lem: small r}
Let $1\leq r< q$, $l\ge 0$, $d=q^lr$ and let $q^{k-1}\leq d$. Then \ $N_q(q^k,d)\ge g_q(k,d)$.
\end{lem}
\begin{proof}
Since $1\leq r<q$, the hypothesis $q^{k-1}\leq d$ is equivalent to $k-1\leq l$, hence $q^{k-1}\mid d$ and we can apply Proposition \ref{prop: plotkin equal griesmer}.  
\end{proof}
\begin{prop}\label{prop: griesmer qlr}
Let $1\leq r<q$ and $l\ge0$. Then \ $S_q(k,q^lr)\ge g_q(k,q^lr)$.
\end{prop}
\begin{proof}
Due to Theorem \ref{thm: minimum k q} we only need to prove that the Griesmer bound is true for all choices of $k$ such that $q^{k-1}\leq d$.
Then we can use Lemma \ref{lem: small r}, which ensures that all such codes respect the Griesmer bound.
\end{proof}
\begin{cor}\label{cor: griesmer 2l}
Let $q=2$ and $l\ge 0$. Then $S_2(k,2^l)\ge g_2(k,2^l)$.
\end{cor}
\begin{proof}
It follows directly from Proposition \ref{prop: griesmer qlr}, with $r=1$.
\end{proof}
\subsection{The case $q=2$, $d=2^r - 2^s$}\label{subsec: 2r2s}
In this section we prove that the Griesmer bound holds for all binary systematic codes whose distance is the difference of two powers of $2$. We need the following lemmas.
\begin{lem}\label{lem: g(r+1)=2g(r)}
Let $r\ge 0$ and let $k\leq r+1$. Then $$g_2(k,2^{r+1}) = 2 g_2(k,2^r) .$$
\end{lem}
\begin{proof}
The hypothesis $k\leq r+1$ implies that for any $i\leq k-1$, both $\left\lceil\frac{2^{r+1}}{2^i}\right\rceil=\frac{2^{r+1}}{2^i}$ and $\left\lceil\frac{2^{r}}{2^i}\right\rceil=\frac{2^{r}}{2^i}$. Therefore
$$
g_2(k,2^{r+1})=\sum_{i=0}^{k-1}\left\lceil
\frac{2^{r+1}}{2^i}
\right\rceil
=\sum_{i=0}^{k-1}
\frac{2^{r+1}}{2^i}
=2\sum_{i=0}^{k-1}
\frac{2^{r}}{2^i}
=2\sum_{i=0}^{k-1}
\left\lceil\frac{2^{r}}{2^i}\right\rceil
=
2g_2(k,2^r)
$$
\end{proof}
\begin{lem}\label{lem: increase g(k,d)}
Let $l\ge 0$ be the maximum integer such that $2^l$ divides $d$. Then
\begin{equation}\label{eq: inrease g(k,d)}
g_2(k,d+1)=g_2(k,d)+\min(k,l+1),
\end{equation}
\end{lem}
\begin{proof}
Clearly $d=2^lr$, where $r$ is odd, and the Griesmer bound becomes
\begin{equation}\label{eq: griesmer increase distance}
g_2(k,d+1)=\sum_{i=0}^{k-1}\left\lceil\frac{2^lr+1}{2^i}\right\rceil .
\end{equation}

We consider first the case $k\leq l+1$, 
and we observe that for each $i$ we have
$$
\left\lceil\frac{2^lr+1}{2^i}\right\rceil=
\frac{2^lr}{2^i}+ \left\lceil\frac{1}{2^i}\right\rceil=
\frac{2^lr}{2^i}+1=
\left\lceil\frac{2^lr}{2^i}\right\rceil+1.
$$
Therefore 
\begin{equation}\label{eq: k<l+1}
g_2(k,d+1)
=\sum_{i=0}^{k-1}\left(\left\lceil\frac{2^lr}{2^i}\right\rceil+1\right)
=g_2(k,d)+k.
\end{equation}

If $k>l+1$ we can split the sum \eqref{eq: griesmer increase distance} in the two following sums:
\begin{equation}\label{eq: split sum}
g_2(k,d+1)=\left(\sum_{i=0}^{l}\left\lceil\frac{2^lr+1}{2^i}\right\rceil\right)
+
\left(\sum_{i=l+1}^{k-1}\left\lceil\frac{2^lr+1}{2^i}\right\rceil\right).
\end{equation}
For the first sum we make use of the same argument as above, while for the second sum we observe that $i>l$, which implies
$$
\left\lceil\frac{2^lr+1}{2^i}\right\rceil=\left\lceil\frac{2^lr}{2^i}\right\rceil.
$$
Putting together the two sums, equation \eqref{eq: split sum} becomes
\begin{equation}\nonumber
g_2(k,d+1)=\left(\sum_{i=0}^{l}\left\lceil\frac{2^lr}{2^i}\right\rceil+l+1\right)
+
\left(\sum_{i=l+1}^{k-1}\left\lceil\frac{2^lr}{2^i}\right\rceil\right)
=
\sum_{i=0}^{k-1}\left\lceil\frac{2^lr}{2^i}\right\rceil+l+1,
\end{equation}
and the term on the right-hand side is $g_2(k,d)+l+1$. Together with \eqref{eq: k<l+1} this concludes the proof.
\end{proof}


\begin{lem}\label{lem: gk2r-gk2r2s}
Let $k$, $r$ and $s$ be integers such that $r> s$ and $k> s+1$. Then
$$
g_2(k,2^r)-g_2(k,2^r-2^s)=2^{s+1}-1 .
$$
\end{lem}
\begin{proof}
For any $d'$ in the range $2^r-2^s\leq d'<2^r$ we can apply Lemma \ref{lem: increase g(k,d)}, observing that $d'=2^l\rho$ where $\rho\nmid d'$ and $l\leq s$, which implies $k>l+1$. 
In particular we observe that $d'=2^r-\delta$ for a certain $\delta\leq2^s$, and since $2^l$ has to divide both $2^r$ and $\delta$ it follows that $l$ depends only on the latter. For a fixed $\delta$ we denote with $l_{\delta}$ the corresponding exponent.\\
From Lemma \ref{lem: increase g(k,d)} we obtain
$$
g_2(k,2^r-\delta+1)=g_2(k,2^r-\delta)+l_{\delta}+1.
$$
Applying it for all distances from $2^r-2^s$ to $2^r$ we obtain
\begin{equation}\label{eq: difference g with S}
g_2(k,2^r)-g_2(k,2^r-2^s)=\sum_{\delta=1}^{2^s}\left(l_{\delta}+1\right)=\sum_{\delta=1}^{2^s}l_{\delta}+2^s .
\end{equation}
For each value of $s$, we call $L_s=(l_1,\ldots,l_{2^s})$ the sequence of integers $\{l_{\delta}\}$ that appear in equation \eqref{eq: difference g with S}, and with $T_s$ the sum itself, so that we can write equation \eqref{eq: difference g with S} as
$$
g_2(k,2^r)-g_2(k,2^r-2^s)=T_s+2^s .
$$
In the following we will prove that $T_s=2^s-1$.
First, we show that $L_{s}=(l_1,\ldots,l_{2^{s}})$ is equal to
$$
(l_1,\ldots,l_{2^{s-1}},l_1,\ldots,l_{2^{s-1}-1},l_{2^{s-1}}+1) ,
$$
namely the first $2^{s-1}$ terms are exactly the sequence $L_{s-1}$, while the second half of the sequence is itself equal to $L_{s-1}$ with the exception of the last term, which is incremented by $1$.\\
The fact that the first $2^{s-1}$ elements of $L_s$ are the elements of $L_{s-1}$ follows directly from the definition of $L_s$, since $l_{\delta}$ is the largest integer such that $2^{l_{\delta}}\mid \delta$. For the same reason, $l_{2^s}=l_{2^{s-1}}+1$.
We take now an element in the second half of $L_s$, which can be written as $l_{2^{s-1}+\bar{\delta}}$, for a certain $1\leq \bar{\delta}\leq 2^{s-1}$. Using the same argument as before, the integer $l_{2^{s-1}+\bar{\delta}}$ depends only on $\bar{\delta}$ and is equal to $l_{\bar{\delta}}$. 
\\
To provide some examples, we have
\begin{center}
\begin{tabular}{|c|c|c|c|c|}
\hline
$s$ & 1 & 2 & 3 & 4\\
\hline
$L_s$ & (0,1) & (0,1,0,2) & (0,1,0,2,0,1,0,3) &  (0,1,0,2,0,1,0,3,0,1,0,2,0,1,0,4) \\
\hline
\end{tabular}
\end{center}
From the properties of $L_s$ it follows that $T_s=2T_{s-1}+1$.
Using induction on $s$, with first step $T_1=2^1-1$, we now prove our claim 
$T_s=2^s-1$: 
if $T_{s-1}=2^{s-1}-1$, then 
\begin{equation}\label{eq: value S s}
T_s=2T_{s-1}+1=2\left(2^{s-1}-1\right)+1=2^{s}-1 .
\end{equation}
Putting together equations \eqref{eq: difference g with S} and \eqref{eq: value S s} we obtain 
\begin{equation}\nonumber
g_2(k,2^r)-g_2(k,2^r-2^s)=2^{s}-1+2^s=2^{s+1}-1 .
\end{equation}
\end{proof}


\begin{lem}\label{lem: g < 2^r}
If $k\leq r$, then $g_2(k,2^r)<2^{r+1}$.
\end{lem}
\begin{proof}
Due to $k\leq r$, for $i<k$ it holds $\left\lceil\frac{2^{r}}{2^i}\right\rceil=\frac{2^{r}}{2^i}$. We can write the Griesmer bound as
$$
g_2(k,2^r)=\sum_{i=0}^{k-1}\frac{2^r}{2^i}=2^r\sum_{i=0}^{k-1}\frac{1}{2^i}<2^r\cdot 2.
$$
\end{proof}

\begin{thm}\label{thm:griesmer_2r2s}
Let $r$ and $s$ be integers such that $r>s\ge 1$ and let $d=2^r-2^s$. Then $S_2(k,d)\ge g_2(k,d)$.
\end{thm}
\begin{proof}
If $r=s+1$, then $2^r-2^s=2^s$, hence we can apply Corollary \ref{cor: griesmer 2l} and our claim holds. Therefore we can assume $r\ge s+2$ in the rest of the proof.\\
\indent Our proof is by contradiction, by supposing that $S_2(k,2^r-2^s)< g_2(k, 2^r-2^s)$, 
i.e. the Griesmer bound does not hold for some $(n,2^k,d)_2$ systematic code $C$, with $d=2^r-2^s$ and $n = S_2(k,d)$. Due to Theorem \ref{thm: minimum k q}, we can assume that $k<1+\log_2d$ and so $k\leq r$.
\\
We call $m$ the ratio $n/d$, which in the case of $C$ is
\begin{equation}\label{eq: inequality m}
m=\frac{S_2(k,2^r-2^s)}{2^r-2^s}
\leq \frac{g_2(k,2^r-2^s)-1}{2^r-2^s}
\end{equation}
We claim that 
\begin{equation}\label{eq: claim thm difference of powers of 2}
m<\frac{g_2(k,2^r)}{2^r} .
\end{equation}
First we observe that since $k\leq r$, then 
\begin{equation}\nonumber
\frac{g_2(k,2^r)}{2^r}=
\sum_{i=0}^{k-1}\frac{1}{2^i}=
2\left(1-\frac{1}{2^k}
\right).
\end{equation}
We consider now the ratio $m$:
\begin{equation}\label{eq: m}
m \leq \frac{g_2(k,2^r-2^s)-1}{2^r-2^s}=
\frac{1}{2^r-2^s}\sum_{i=0}^{k-1}\left\lceil
\frac{2^{r}-2^s}{2^i}
\right\rceil
-\frac{1}{2^r-2^s}
\end{equation}
\indent We consider first the case $k\leq s+1$, and we can write \eqref{eq: m} as
\begin{equation}\nonumber
m < 
\frac{1}{2^r-2^s}\sum_{i=0}^{k-1}
\frac{2^{r}-2^s}{2^i}=
\sum_{i=0}^{k-1}\frac{1}{2^i}=
2\left(1-\frac{1}{2^k}
\right),
\end{equation}
so in this case $m<\frac{g_2(k,2^r)}{2^r}$, which is exactly claim \eqref{eq: claim thm difference of powers of 2}. \\
\indent We consider now the case $k\ge s+2$. To prove \eqref{eq: claim thm difference of powers of 2}, we prove that the term on the right-hand side of inequality \eqref{eq: inequality m} is itself less than $\frac{g_2(k,2^r)}{2^r}$, and we write this claim in the following equivalent way: 
$$
2^r(g_2(k,2^r-2^s)-1)<(2^r-2^s)g_2(k,2^r).
$$
Rearranging the terms we obtain
\begin{equation}\label{eq: equivalent claim m}
2^sg_2(k,2^r)<2^r(g_2(k,2^r)-g_2(k,2^r-2^s)+1)=2^r\cdot 2^{s+1} ,
\end{equation}
where the equality on the right hand side is obtained from Lemma \ref{lem: gk2r-gk2r2s}. Hence
\begin{equation}\nonumber
g_2(k,2^r)<2^{r+1} ,
\end{equation}
and this is always true provided $k\leq r$, as shown in Lemma \ref{lem: g < 2^r}. This concludes the proof of claim \eqref{eq: claim thm difference of powers of 2}.
\\
\indent We now consider the $(tn,2^k,td)_2$ systematic code $C_t$ obtained by repeating $t$ times the code $C$. We remark that the value $m$ can be thought of as the slope of the line $\dd(C_t)\mapsto \len(C_t)$, and we proved that $m<\frac{g_2(k,2^r)}{2^r}$. Since $k\leq r$ we can apply Lemma \ref{lem: g(r+1)=2g(r)}, which ensures that $g_2(k,2^{r+b})=2^bg_2(k,2^r)$, namely the Griesmer bound computed on the powers of $2$ is itself a line, and its slope is strictly greater than $m$. Due to this, we can find a pair $(t,\;b)$ such that the code $C_t$ is an $(tn,2^k,td)_2$ systematic code where
\begin{enumerate}
 \item  $td > 2^b$,
 \item  $tn < g_2(k, 2^b)$. 
\end{enumerate}
We can now apply Lemma \ref{lem: reducing d} to $C_t$, and find a systematic code with length $tn$ and distance equal to $2^b$, which means we have an $(tn,k,2^b)_2$ systematic code for which the length is $tn < g_2(k,2^b)$. 
This however contradicts Corollary \ref{cor: griesmer 2l}, hence for each $k\leq r$ we have 
$$
S_2(k,2^r-2^s)\ge g_2(k,2^r-2^s).
$$
\end{proof}
\begin{cor}\label{cor: odd distances}
Let $r$ and $s$ be integers such that $r>s\ge 1$, and let $d$ be either $2^s-1$ or $2^r-2^s-1$. Then $S_2(k,d)\ge g_2(k,d)$.
\end{cor}
\begin{proof}
We prove it for $d=2^r-2^s-1$, and the same argument can be applied to $d=2^s-1$  by applying Corollary \ref{cor: griesmer 2l} instead of Theorem \ref{thm:griesmer_2r2s}.\\
Suppose by contradiction $S_2(k,d)< g_2(k,d)$, i..e. there exists an $(n,k,d)_2$ systematic code for which 
\begin{equation}\label{eq: proof d odd}
n<g_2(k,d).
\end{equation}
We can extend such a code to an $(n+1,k,d+1)_2$ systematic code $C$ by adding a parity-check component to each codeword. Then $C$ has distance $\dd(C)=d+1 = 2^r-2^s$, so we can apply Theorem \ref{thm:griesmer_2r2s} to it, finding
$$
n+1\ge g_2(k,d+1) .
$$
Observe that $d$ is odd, so applying Lemma \ref{lem: increase g(k,d)} we obtain
$$n+1\ge g_2(k,d+1) =g_2(k,d)+1\quad \implies \quad n\ge g_2(k,d),$$
which contradicts \eqref{eq: proof d odd}.
\end{proof}
%
\section{Versions of the Griesmer bound holding for nonlinear codes}
\label{sec:weak_griesmer}

In this section we collect some minor results which can be seen as bounds on the length of systematic codes, useful for a better understanding of the structure of such codes. An example of codes meeting these bounds are Simplex codes, while Preparata codes and Kerdock codes are close to these bounds. We will discuss some properties of Simplex codes in Section \ref{sec:counterexample}. We recall that Preparata codes are $\left(2^{2m}, 2^{2^{2m}-4m},6\right)_2$ systematic codes while Kerdock codes are $\left(2^{2m}, 2^{4m}, 2^{2m-1}-2^{m-1}\right)_2$ systematic codes, both with $m\ge 2$. For $m=2$ the two codes are both equivalent to the Nordstrom-Robinson code, which is a $(16,2^8,6)_2$ systematic binary code meeting the bound in Corollary \ref{bound: B binary}.
\\
In Table \ref{table: improvements} there is a (not exhaustive) list of parameters $n, d$ for which the binary bound in Equation \eqref{eq: bound B binary} outperforms some known bounds, such as the Singleton Bound, the Elias bound, the Hamming Bound and the Johnson Bound.

\subsection{An improvement of the Singleton bound}\label{subsec: binary bound}
For systematic binary codes we can improve the Singleton bound as follows.
\begin{prop}[Bound A]\label{prop: singleton improved}
$$
S_2(k,d)\ge k+\left\lceil\frac{3}{2}d\right\rceil-2.
$$
\end{prop}
\begin{proof}
We will proceed in a similar manner as in the proof of the Griesmer bound. \\
We consider a binary $(n=S_2(k,d),2^k,d)_2$ systematic code $C$. We consider the set $S$ of all codewords whose weight in their systematic part is $1$. Let $c$ be a codeword in this set with minimum weight:
\begin{equation}\label{eq: c has minimum weight}
\mathrm{w}(c)=\min_{x\in S}\{\mathrm{w}(x)\} .
\end{equation}
Since we can always assume without loss of generality that the zero codeword belongs to $C$, the weight of $c$ is at least $d$, and we denote it with $d+\Delta$, $\Delta\ge 0$.
We also assume that the non-zero coordinates of $c$ are the first $d+\Delta$, and that the first coordinate is the only non-zero systematic coordinate of $c$.\\
We construct a code $C'$ by shortening $C$ in the first coordinate and by puncturing it in the remaining $d+\Delta-1$ first coordinates. Since the shortening involves a systematic coordinate and the puncturing does not affect the systematic part of $C$, $C'$ is an $(n-d-\Delta,2^{k-1},d')_2$ systematic code.\\
We consider now a codeword $u$ in $C'$, such that $u$ has weight $1$ in its systematic part. Then there exists a vector $v\in\left(\mathbb{F}_2\right)^{d+\Delta}$ such that the concatenation $(v\mid u)$ belongs to $C$. We remark that even though there may be many vectors satisfying this property, we can choose $v$ such that its first component is $0$, and this choice is unique. Therefore $(v\mid u)\in S$, and due to equation \eqref{eq: c has minimum weight} 
\begin{equation}\label{eq: weight v|u}
\mathrm{w}(v\mid u)=\mathrm{w}(v)+\mathrm{w}(u)\ge d+\Delta .
\end{equation}
Moreover, we can also bound the distance of $(v\mid u)$ from $c$ as follows:
\begin{equation}\label{eq: bound distance from c}
\mathrm{d}(c,v\mid u)=d+\Delta-\mathrm{w}(v)+\mathrm{w}(u)\ge d
\end{equation}
Summing together the inequalities \eqref{eq: weight v|u} and \eqref{eq: bound distance from c} we have
$$
d+\Delta+2\mathrm{w}(u)\ge 2d+\Delta,
$$
from which it follows that
$$
\mathrm{w}(u)\ge\frac{d}{2} .
$$
Since $u$ has weight $1$ in its systematic part, it means that its weight in the non-systematic part is at least $\frac{d}{2}-1$. So $u$ has $k-1$ systematic coordinates and at least $\frac{d}{2}-1$ non-systematic coordinates:
$$
\len(C')\ge (k-1)+\left(\frac{d}{2}-1\right).
$$
Since the length of $C'$ is $n-d-\Delta$ we have
$$
n-d-\Delta\ge k+\frac{d}{2}-2 ,
$$
or equivalently
$$
n\ge k+\frac{3d}{2}-2+\Delta
$$
which implies the bound.
\end{proof}
\subsection{Consequences of Proposition \ref{prop: griesmer qlr}}\label{subsec: consequences qrl}
We derive from Proposition \ref{prop: griesmer qlr} a version of the Griesmer bound holding for any systematic code.
\begin{rem}\label{rem: exist l r}
For any $d$, there exist $1\leq r < q$ and $l\ge 0$ such that
\begin{equation}\label{eq: r l}
q^lr\leq d<q^l(r+1)\leq q^{l+1}
\end{equation}
Thus $l$ has to be equal to $\left\lfloor \log_qd\right\rfloor$, and from inequality \eqref{eq: r l} we obtain
$d/q^l-1< r\leq d/q^l$,
namely $r=\left\lfloor d/q^l \right\rfloor.$
\end{rem}
\begin{cor}[Bound B]\label{cor: griesmer qlr systematic}
Let $l=\left\lfloor\log_qd\right\rfloor$ and $r=\left\lfloor d/q^l \right\rfloor$. Then
$$
S_q(k,d)\ge d+\sum_{i=1}^{k-1}\left\lceil \frac{q^lr}{q^i} \right\rceil.
$$
\end{cor}
\begin{proof}
We denote $s=d-q^lr$. We remark that $s\leq n-k$, and so there are at least $s$ non-systematic coordinates.
With this notation, let $C$ be an $(n,q^k,q^lr+s)_q$ systematic code. 
We build a new systematic code $C_s$ by puncturing $C$ in $s$ non-systematic coordinates. $C_s$ has parameters $(n-s, q^k, d_s)_q$, for a certain $q^lr\leq d_s\leq q^lr+s$.\\
If $q^lr\neq d_s$, we can apply Lemma \ref{lem: reducing d}, in order to obtain another code $\bar{C}$, 
so that we have an $(n-s,q^k,q^lr)_q$ systematic code. 
Due to Remark \ref{rem: exist l r}, it holds $1\leq r < q$, so we can apply Proposition \ref{prop: griesmer qlr} to $\bar{C}$. 
We find
$n-s\ge \sum_{i=0}^{k-1}\left\lceil \frac{q^lr}{q^i} \right\rceil$,
hence 
$n\ge \sum_{i=0}^{k-1}\left\lceil \frac{q^lr}{q^i} \right\rceil+s$.
We finally remark that for $i=0$ we have $\left\lceil \frac{q^lr}{q^i}\right\rceil=q^lr$, and by adding $s$ we obtain exactly $d$. 
So
$ n\ge d+\sum_{i=1}^{k-1}\left\lceil \frac{q^lr}{q^i} \right\rceil$.
\end{proof}
We also derive a similar bound for binary codes, whose proof relies on Theorem \ref{thm:griesmer_2r2s} instead of Proposition \ref{prop: griesmer qlr}.
\begin{cor}[Bound B, binary version]\label{bound: B binary}
Let $C$ be an $(n,2^k,d)_2$ systematic code with $d$ even. Let $r$ and $s$ be the smallest integers such that $2^r-2^s\leq d< 2^r$, namely $r=\lceil \log_2(d+1)\rceil$ and $s=\lceil\log_2(2^r-d)\rceil$. Then
\begin{equation}\label{eq: bound B binary}
n\ge d+\sum_{i=1}^{k-1}\left\lceil \frac{2^r-2^s}{2^i} \right\rceil .
\end{equation}
\end{cor}
\begin{proof}
It follows directly from Theorem \ref{thm:griesmer_2r2s}.
\end{proof}
In Table \ref{table: improvements} we list some values $n$ and $d$ for which Bound B in Proposition \ref{bound: B binary} outperforms known bounds. 
The first two rows are respectively $n$ and $d$. In the third row, we have the maximum combinatorial dimension allowed by the Elias Bound (EB). The last row is the bound obtained using Equation \eqref{eq: bound B binary}. 
We did not list other bounds in the table since for these values $n$ and $d$ the combinatorial dimensions obtained from the Hamming bound, the Singleton bound and the Johnson bound are at least equal to the one obtained from the Elias bound, while the Plotkin bound cannot be applied.

\begin{table}
\begin{center}
\begin{tabular}{|c|c|c|c|c|c|c|c|c|c|c|c|c|}
\hline
n&26&28&28&30&32&33\\
\hline
d&12&12&14&14&16&16\\
\hline
Elias bound &8&10&6&8&7&8\\
\hline
Bound B &7&9&5&7&6&7\\
\hline
\end{tabular}
\vspace{.1cm}
\caption{Bound B}\label{table: improvements}
\end{center}
\end{table}

\subsection{Consequences of Corollary \ref{cor: griesmer 1} }\label{subsec: prop gries 1}
The following two bounds can be applied to nonlinear codes.
\begin{prop}[Bound C]\label{prop: griesmer 2}
Let $l$ be the maximum integer such that $q^l$ divides $d$, and let $h=\min\left(k-1,l\right)$.
Then
$$
S_q(k,d)\ge N_q(q^k,d) \ge \sum_{i=0}^{h}\left\lceil \frac{d}{q^i}\right\rceil .
$$
\end{prop}
\begin{proof}
First, notice that $d=q^lr,\quad q\nmid r$.
If $(k-1) \mid l$, we apply Lemma \ref{lem: small r}. 
Otherwise $h=l$, and $d$ is not divisible for higher powers of $q$, and the laast term of the sum is $\frac{d}{q^l}$.
\end{proof}
We remark that, if there exists an $(n,M,d)_q$ code, then there exists also an $(n,q^k,d)_q$ code, with $q^k\leq M$. By Proposition  \ref{prop: griesmer 2} we have
$$
N_q(M,d) \ge \sum_{i=0}^{h}\left\lceil \frac{d}{q^i}\right\rceil .
$$
\section{Classification of optimal binary codes with $4$ codewords}\label{sec: optimal 4 words}
In the previous sections we have focused our attention on the distance, proving that for particular choices of $d$ the length of optimal systematic codes is at least the Griesmer bound, for each possible dimension. In the next sections we deal with the task of characterize optimal systematic codes depending on their dimension. In particular in this section we  prove that all optimal binary codes with $4$ codewords are linear codes, and so they are systematic codes. We recall our convention $0\in C$. A first version of this proof appeared in \cite{CGC-cd-phdthesis-ele}.

\begin{lem}\label{lem: equivalence 4 codewords}
$N_2(4,d)=S_2(2,d)=L_2(2,d)$.
\end{lem}
 \begin{proof}
 We are going to show that $N_2(4,d)\ge L_2(2,d)$, and then Remark \ref{rem: inequalities} will conclude the proof.
 \\
 Let $C=\{c_0,\,c_1,\,c_2,\,c_3\}$ be an optimal $(n,4,d)_2$ code, i.e. $n=N_2(4,d)$, and we assume without loss of generality that $c_0$ is the zero codeword. The weights of $c_1$ and $c_2$ are at least $d$, and their distance is $d(c_1,c_2)=\mathrm{w}(c_1+c_2)\ge d$. Therefore the linear code generated by $c_1$ and $c_2$ have the same minimum distance as $C$, and it follows that $n\ge L_2(2,d)$.
 \end{proof}
A consequence of Lemma \ref{lem: equivalence 4 codewords} is that the Griesmer bound holds for all binary (nonlinear) codes with $4$ codewords. Furthermore, using the argument of the proof of Lemma \ref{lem: equivalence 4 codewords} we can build (binary optimal) linear codes starting from nonlinear ones. This construction is however not necessary, as explained in the following theorem.
\begin{thm}\label{thm: optimal are linear}
Let $C$ be an optimal $(n,4,d)_2$ code. Then $C$ is a linear code.
\end{thm}
\begin{proof}
As in the proof of Lemma \ref{lem: equivalence 4 codewords}, we assume that $c_0$ is the zero codeword. If $C$ is not linear, then there exists at least a position $i$ for which the $i$-th coordinate of $c_3$ is different from the $i$-th coordinate of $c_1+c_2$. Looking at the $i$-th components of the four codewords as a vector $v$ in $\left(\mathbb{F}_2\right)^4$ we claim to have only two possibilities: either $\mathrm{w}(v)=1$ or $\mathrm{w}(v)=3$. In fact, $\mathrm{w}(v)=0$ implies that $C$ is not optimal, $\mathrm{w}(v)=4$ contradicts the fact that $c_0\in C$ and $\mathrm{w}(v)=2$ contradicts the choice of $i$. Without loss of generality we can assume that we are in one of the following two cases:
\begin{equation}\nonumber
v=\left(
0,0,0,1
\right)
 \qquad\mathrm{or}\qquad 
v=\left(
0,1,1,1
\right)
\end{equation}
\indent We start from the first case, namely $\mathrm{w}(v)=1$, and we consider the $[n,2,d]_2$ linear code $\bar{C}$ generated by $c_1$ and $c_2$. Clearly, all codewords in $\bar{C}$ have the $i$-th component equal to zero. Then we can puncture $\bar{C}$, obtaining a $[n-1,2,d]_2$ linear code, contradicting the fact that $C$ is optimal.
\\
\indent We consider the second case, namely $\mathrm{w}(v)=3$. We consider the code $\tilde{C}$ obtained by adding $c_3$ to each codeword in $C$. $\tilde{C}$ is an optimal code with the same parameters as $C$, and the zero codeword still belongs to the code. However what we obtain looking at the $i$-th coordinate is a vector of weight $1$, and we can use the same argument as in the first case. 
\end{proof}
\begin{cor}
The Griesmer bound holds for binary codes with $4$ codewords. Furthermore
\begin{equation}\nonumber
N_2(4,d)=S_2(4,d)=L_2(2,d)=
\left\{
\begin{array}{ll}
\frac{3}{2}d,&\mathrm{if}\;d\;\mathrm{is}\;\mathrm{even}\\
\frac{3}{2}(d+1)-1,&\mathrm{if}\;d\;\mathrm{is}\;\mathrm{odd}
\end{array}
\right.
\end{equation}
\end{cor}
\begin{proof}
The fact that the Griesmer bound holds for all codes of size $4$ follows directly from Lemma \ref{lem: equivalence 4 codewords} or Theorem~\ref{thm: optimal are linear}. This implies that
$$
N_2(4,d)\ge d+\left\lceil
\frac{d}{2}
\right\rceil
$$
We consider $d$ even, so that the previous equation is $N_2(4,d)=\frac{3}{2}d$. It is straightforward to exhibit a $\left[\frac{3}{2}d,2,d\right]_2$ linear code $C$, and this concludes the proof in the case of $d$ even. On the other hand, by puncturing $C$ we obtain a $\left[\frac{3}{2}d-1,2,d-1\right]_2$ linear code, which proves the case of odd distance.
\end{proof}


\section{On the structure of optimal binary codes with 8 codewords}\label{sec: code structure 8 codewords}

We consider in this section optimal codes with 8 codewords. First we prove that for these codes the Plotkin bound and the Griesmer bound coincide, implying that the Griesmer bound actually holds also for them. 

\begin{prop}\label{prop: griesmer eq plotkin}
For any $d$, $N_2(8,d)\ge g_2(3,d)$, namely
\begin{equation}
N_2(8,d)\ge\left\{
\begin{array}{ll}
7h ,& \mathrm{if}\;d=4h\\ 
7h+3 ,& \mathrm{if}\;d=4h+1\\ 
7h+4 ,& \mathrm{if}\;d=4h+2\\ 
7h+6 ,& \mathrm{if}\;d=4h+3 
\end{array}
\right.
.
\end{equation}
\end{prop}
\begin{proof}
Let us consider an $(N_2(8,d),8,d)_2$ code $C$.
Let $h=\left\lfloor\frac{d}{4}\right\rfloor$. There are four cases for $d$:
$$
d=4h,\qquad
d=4h+1,\qquad
d=4h+2,\qquad
d=4h+3 .
$$
\indent We start with the case $d=4h$ (so $h\ge1$),
for which
$$
g_2(3,4h)=\sum_{i=0}^2\left\lceil\frac{4h}{2^i}\right\rceil=7h .
$$
On the other hand, by the Plotkin bound we have
$$
N_2(8,d)\ge \min\left\{n\in \mathbb{N} \mid 8\leq 2\left\lfloor \frac{4h}{8h-n}\right\rfloor\right\} .
$$
Assuming $n< 7h$, we have $8h-n> h$. 
This implies that $$4>\frac{4h}{8h-n} ,$$ which contradicts our hypothesis and shows that the Griesmer bound and the Plotkin bound coincide.
\\
\indent In the case of $d=4h+2$,
$$
g_2(3,4h+2)=\sum_{i=0}^2\left\lceil\frac{4h+2}{2^i}\right\rceil=(4h+2)+(2h+1)+(h+1)=7h+4 .
$$
By the Plotkin bound
$$
\frac{4h+2}{8h+4-N_2(8,d)}
$$
which is equivalent to
$
N_2(8,d)\ge 7h+4.
$
\\
\indent In the case of $d=4h+1$, 
$$
8\leq 2\left\lfloor \frac{4h+2}{8h+3-N_2(8,d)}\right\rfloor ,
$$
hence $N_2(8,d)\ge7h+3$.
\\
\indent Finally, in the case of $d=4h+3$, by the same computation as above we obtain that $N_2(8,d)\ge7h+6$.
\end{proof}

\begin{thm}\label{thm: construction linear}
For any $d$, $L_2(3,d)=g_2(3,d) .$
\end{thm}
\begin{proof}
We consider the following three binary matrices:
$$
I_3=
\begin{bmatrix}
1&0&0\\
0&1&0\\
0&0&1
\end{bmatrix}
, \;1_3=
\begin{bmatrix}
1\\
1\\
1
\end{bmatrix},\;
N_3=
\begin{bmatrix}
0&1&1\\
1&0&1\\
1&1&0
\end{bmatrix}.
$$
We remark that the code generated by $I_3$ (resp. $\left[\;I_3\; \middle|\; 1_3\; \right]$ and $\left[\;I_3\; \middle|\; N_3\; \right]$) is a $[3,3,1]_2$ (resp. a $[4,3,2]_2$ and a $[6,3,3]_2$) linear code. These codes meet the Griesmer bound. We denote with $G_3$ the matrix $\left[\;I_3\; \middle|\; N_3\; \middle|\; 1_3\;\right]$, i.e.
$$
G_3=
\begin{bmatrix}
1&0&0&0&1&1&1\\
0&1&0&1&0&1&1\\
0&0&1&1&1&0&1
\end{bmatrix}.
$$
The code generated by $G_3$ is a $[7,3,4]_2$ linear code, which again attains the Griesmer bound. Thus, $L_2(3,d)=g_2(3,d)$ for $1\leq d\leq 4$.\\
\indent
 Let $d=4h$. 
We denote with $G_{3,h}$ the $3\times 7h$ matrix obtained by repeating $h$ times the matrix $G_3$. The code generated by $G_{3,h}$ is a $[7h,3,4h]_2$ linear code, which attains the Griesmer bound.
\\
\indent For the other three cases, we consider the matrices
$$
\left\{
\begin{array}{l}
\left[\;G_{3,h}\; \middle|\; I_3\; \right]\\ \\
\left[\;G_{3,h}\; \middle|\; I_3\; \middle|\; 1_3\;\right]\\ \\
\left[\;G_{3,h}\; \middle|\; I_3\; \middle|\; N_3\;\right] ,
\end{array}
\right. 
$$
that generate, respectively, a $[7h+3,3,4h+1]_2$, a $[7h+4,3,4h+2]_2$ and a $[7h+6,3,4h+3]_2$ linear code, each attaining the Griesmer bound.
\end{proof}
Propositions \ref{prop: griesmer eq plotkin} and Theorem \ref{thm: construction linear} imply the following corollary.
\begin{cor}\label{cor: optimal length}
For any $d$, $N_2(8,d)=S_2(3,d)=L_2(3,d)$, and 
\begin{equation}
N_2(8,d)=\left\{
\begin{array}{ll}
7h ,& \mathrm{if}\;d=4h\\ 
7h+3 ,& \mathrm{if}\;d=4h+1\\ 
7h+4 ,& \mathrm{if}\;d=4h+2\\ 
7h+6 ,& \mathrm{if}\;d=4h+3
\end{array}
\right.
\end{equation}
\end{cor}

\section{Counterexamples to the Griesmer bound: a family of optimal systematic binary codes}
\label{sec:counterexample}
In previous sections we identified several sets of parameters for which the Griesmer bound holds in the systematic case. In this section we focus our attention on binary systematic (nonlinear) code for which the Griesmer bound does not hold. It is known that there exist pairs $(k,d)$ for which $N_2(2^k,d)<g_2(k,d)$, but it has not been clear whether the same is true for systematic codes. In this section we construct a family of optimal systematic nonlinear codes contradicting the Griesmer bound. 
In \cite{CGC-cod-art-Levenshtein64application}, Levenshtein has shown that if Hadamard matrices of certain orders exist, 
then the binary codes obtained from them meet the Plotkin bound. 
Levenshtein's method to construct such codes can be found also in the proof of Theorem 8 of \cite[Ch.~2,\S 3]{CGC-cd-book-macwilliamsTOT}.
In particular, given a Hadamard matrix of order $2^k+4$, it is possible to construct a  $(2^k+3,2^k,2^{k-1}+2)_2$ code $D_k$.
We recall that binary codes attaining the Plotkin bound are equidistant codes.
\begin{defn}
A code $C$ is called an \textit{equidistant} code if any two codewords have the same distance $d$.
\end{defn}
We consider now the family of binary simplex codes $\mathcal{S}_k$, which can be defined as the codes generated by the $k\times \left(2^{k}-1\right)$ matrices whose columns are all the non-zero vectors of $\left(\mathbb{F}_2\right)^k$. Simplex codes are $[2^k-1,k,2^{k-1}]_2$ equidistant codes. 
The following proposition follows directly from the application of the Plotkin bound to codes with size $2^k$ and distance a multiple of $2^{k-1}$.
\begin{prop}\label{prop: min n, d=2k-1h}
Let $h\ge1$ be a positive integer. Then 
\begin{equation}\nonumber
N_2(2^k,2^{k-1}h)\ge \left(2^k-1\right)h .
\end{equation}
\end{prop}
We recall that all $[(2^k-1)h,k,(2^{k-1})h]_2$ codes are equivalent to a sequence of Simplex codes  \cite{CGC-cod-bonisoli1984every}. This fact lead to the following corollary.
\begin{cor}
Let $h\ge 1$, then $N_2\left(2^k,2^{k-1}h\right)=S_2\left(k,2^{k-1}h\right)=L_2\left(k,2^{k-1}h\right)=\left(2^{k}-1\right)h$.
\end{cor}

\par\bigskip
We now make use of $D_k$ and $\mathcal{S}_k$ to construct our claimed family $\mathcal{C}_k$ of optimal systematic codes.
\\
We consider $\mathcal{C}_k$ the $(2^{k+1}+2,2^k,d)_2$ code, with the following properties:
\begin{itemize}
\item puncturing $\mathcal{C}_k$ in the last $2^k+3$ coordinates we obtain $\mathcal{S}_k$;
\item puncturing $\mathcal{C}_k$ in the first $2^k-1$ coordinates we obtain $D_k$.
\end{itemize}
Note that such a code is completely defined.
Since $\mathcal{S}_k$ is a linear code and both $D_k$ and $\mathcal{S}_k$ are equidistant codes, $\mathcal{C}_k$ is an equidistant systematic code with distance $d=2^k+2$.  \\
Applying the Plotkin bound to these parameters, we can see that $\mathcal{C}_k$ is not an optimal code since it has only $2^k$ codewords instead of $2^k+2$. However, if $k\ge 2$, it is optimal as a systematic code, since we can add to it at most two codewords and therefore we cannot increase its dimension while keeping the same distance. On the other hand, by the Griesmer bound we obtain
$$
g_2(k,2^k+2)=\sum_{i=0}^{k-1}\left\lceil\frac{2^k+2}{2^i}\right\rceil
=
\sum_{i=0}^{k-1}2^{k-i}+\sum_{i=0}^{k-1}\left\lceil\frac{2}{2^i}\right\rceil .
$$ 
By direct computation
$
g_2(k,2^k+2)=2^{k+1}+k-1.
$
Since $\len(\mathcal{C}_k)=2^{k+1}+2$, if $k>3$ then $\mathcal{C}_k$ contradicts the Griesmer bound. 
\begin{prop}\label{prop: family Ck}
The family $\mathcal{C}_k$ is a family of optimal systematic equidistant binary codes.
\end{prop}
While in Sections \ref{sec: optimal 4 words} and \ref{sec: code structure 8 codewords} we have shown that codes of dimension $2$ or $3$ cannot contradict the Griesmer bound, by using the family $\mathcal{C}_k$ we can obtain for each possible $k>3$ an optimal systematic code whose length is smaller than the length of any possible linear code with the same dimension and distance, as stated in the following theorem.
\begin{thm}\label{thm: optimal systematic dimension}
Let $k>3$. If there exists a Hadamard matrix of order $2^k+4$, then there exists at least a distance $d$ for which $S_2(k,d)<L_2(k,d)$.
\end{thm}
On the other hand, the family of optimal systematic codes presented in this section have distance $2^k+2$. By puncturing them in a non-systematic component, for each $k>3$, it is possible to construct $(2^{k+1}+1,2^k,2^k+1)_2$ optimal systematic codes contradicting the Griesmer bound. 
Theorem \ref{thm:griesmer_2r2s} and Corollary \ref{cor: odd distances} imply that for $k<3$ optimal systematic codes have to satisfy the Griesmer bound. Putting all together we can state the following theorem.
 \begin{thm}\label{thm: distances counterexamle}
 Let $r$ be a positive integer, and let $d=2^r+1$ or $d=2^r+2$. Then
\begin{enumerate}
\item if $r<3$ then all optimal systematic binary codes with dimension $k$ and distance $d$ have length at least equal to $g_2(k,d)$;
\item if $r>3$, assuming there exists a Hadamard matrix of order $2^k+4$, then $S_2(k,d)<L_2(k,d)$.
\end{enumerate}
 \end{thm} 
 This leaves as open problem the case $r=3$, namely the case of a code whose distance is either $9$ or $10$.

%
%
\section{Conclusions}\label{sec:conclusion}
In this work we provide a collection of results on optimality for systematic codes.
The Griesmer bound is one of the few bounds which can only be applied to linear codes. Classical counterexamples arose from the Levensthein's method for building optimal nonlinear codes, however this method does not provide specific counterexamples for the systematic case. It was therefore not fully understood whether the Griesmer bound would hold for systematic nonlinear codes, or whether there exist families of parameters $(k,d)$ for which the bound could be applied to the nonlinear case. 
\\
As regards nonlinear codes satisfying the Griesmer bound, the main results of our work are Theorem \ref{thm:griesmer_2r2s} and Corollary \ref{cor: odd distances}, in which we prove that the Griesmer bound can be applied to binary systematic nonlinear codes whose distance $d$ is such that
\begin{enumerate}
\item $d=2^r$,
\item $d=2^r-1$,
\item $d=2^r-2^s$, or
\item $d=2^r-2^s-1$.
\end{enumerate}
Moreover, an optimal code with four codewords is linear while with eight codewords attains the Griesmer bound.\\
On the other hand, Theorems \ref{thm: optimal systematic dimension} and \ref{thm: distances counterexamle} prove that the Griesmer bound does not hold in general for systematic codes, and we proved this by explicit construction of the family $\mathcal{C}_k$ of optimal systematic codes. 
In particular, Theorem \ref{thm: optimal systematic dimension} shows that, if $k>3$ is such that Hadamard matrices of order $2^k+4$ exist, then there exists a binary systematic nonlinear code with combinatorial dimension $k$ achieving better error correction capability than any linear code with the same size and length.
Finally, in Section \ref{sec:weak_griesmer} we provide some bounds for systematic codes derived from the Griesmer bound.

\section{Aknowledgements}\label{sec:conclusion}
The first two authors would like to thank their (former) supervisor: the last author. 
The second author would also like to thank Emanuele Bellini for the help in the research about the Griesmer bound, whose partial results were presented at WCC 2015, April 13-17, 2015 Paris, France.

%
%

\bibliographystyle{amsalpha}
\bibliography{RefsCGC}

\end{document}